\documentclass[letterpaper, 10 pt, conference]{ieeeconf}

\IEEEoverridecommandlockouts                         
\title{\LARGE \bf
A Data-Driven Safety Preserving Control Architecture for Constrained Cyber-Physical Systems 
}

\author{Mehran Attar and Walter Lucia% <-this % stops a space
	\thanks{This work was supported in part by the Natural Sciences and Engineering Research Council of Canada (NSERC).}%
	\thanks{Mehran Attar and Walter Lucia are with the Concordia Institute for Information Systems Engineering (CIISE), Concordia University, Montreal, QC, H3G 1M8, CANADA, {\tt\small mehran.attar@concordia.ca}, {\tt\small walter.lucia@concordia.ca}}%	
}

\usepackage{cite}

\usepackage{amsthm}
\usepackage{amsmath}
\usepackage{amsfonts}
\usepackage{mathtools}
\usepackage{amssymb}
\usepackage{graphicx}
\usepackage{textcomp}
\graphicspath{{\images}}
\usepackage{siunitx}
\usepackage{ragged2e}
\usepackage{todonotes}
\usepackage{float}
\usepackage{lipsum}
\theoremstyle{plain}
\usepackage{xhfill}
\usepackage{booktabs}
\usepackage{url}
\usepackage{courier}
\usepackage{algpseudocode}
\usepackage{algorithm}
\usepackage{xr-hyper}
\usepackage{hyperref}

\newcommand{\rr}{\mathop{{\rm I}\mskip-4.0mu{\rm R}}\nolimits}

\theoremstyle{definition}

\newtheorem{lemma}{Lemma}
\newtheorem{proposition}{Proposition}

\theoremstyle{remark}
\newtheorem{remark}{\textbf{Remark}}
\theoremstyle{remark}
\newtheorem{assumption}{Assumption}
\newtheorem{definition}{Definition}

\newtheorem{property}{\textbf{Property}}

\begin{document}

\maketitle 
\thispagestyle{empty}
\pagestyle{empty}
%%%%%%%%%%%%%%%%%%%%%%%%%%%%%%%%%%%%%%%%%%%%%%%%%%%%%%%%%%%%%%%%%%%%%%%%%%%%%%%%
\begin{abstract}
In this paper, we propose a data-driven networked control architecture for unknown and constrained cyber-physical systems capable of detecting networked false-data-injection attacks and ensuring plant's safety. 
In particular, on the controller's side, we design a novel robust anomaly detector that can discover the presence of network attacks using a data-driven outer approximation of the expected robust one-step reachable set. On the other hand, on the plant's side, we design a data-driven safety verification module, which resorts to worst-case arguments to determine if the received control input is safe for the plant's evolution. Whenever necessary, the same module is in charge of replacing the networked controller with a local data-driven set-theoretic model predictive controller, whose objective is to keep the plant's trajectory in a pre-established safe configuration until an attack-free condition is recovered. Numerical simulations involving a two-tank water system illustrate the features and capabilities of the proposed control architecture.
\end{abstract}
%%%%%%%%%%%%%%% SECTION %%%%%%%%%%%%%%%%%%%%%%
\section{Introduction}\label{sec:introduction}

The term Cyber-Physical Systems (CPSs) is used to denote physical systems equipped with communication and computation/control capabilities. Well-known examples of CPSs are smart grids, water distribution systems, and autonomous transportation systems %\cite{mo2011cyber,amin2010stealthy}.
\cite{dibaji2019systems}.
CPSs have the great potential to improve the efficiency and robustness of traditional engineering systems. However, the increased capabilities come with unavoidable concerns about their vulnerability to cyber-attacks.
%\cite{teixeira2015secure}. 
The vulnerability of CPSs to cyber-attacks has motivated engineers to design detection mechanisms to detect the presence of attacks and develop control architecture to preserve the plant's safety. Different passive and active solutions have been proposed in the literature to detect the presence of cyber attacks. In \cite{ghaderi2020blended}, an active detection mechanism has been proposed based on moving target and watermarking ideas to detect stealthy attacks. In \cite{miao2016coding}, a sensor coding method has been proposed to detect stealthy attacks. In \cite{attaractive}, authors have proposed an active detection mechanism by leveraging dimensionality reduction, optimal reconstruction, and a time-varying encoding mechanism to detect stealthy replay and covert attacks. 

As far as attack countermeasures or safe control strategies against networked attacks are concerned, there are a few solutions that have been proposed to overcome the limitations of traditional fault-tolerant control schemes, see, e.g., \cite{franze2023cyber, sun2019resilient, franze2023output} and references therein. Moreover, most of the solutions are developed for unconstrained control systems. One of the first attempts to ensure the safety of constrained CPSs against attacks can be found in \cite{lucia2022supervisor}, where robust set-theoretic arguments are used to design a safe networked control scheme under the assumption that an attack-free scenario can be recovered in a-priori known number of steps. On the other hand, in \cite{gheitasi2022worst}, a safety-preserving architecture is derived by resorting to set-theoretic arguments and a local emergency controller.  
%
%%%%%%%%%%%%%% SECTION %%%%%%%%%%%%%%%%%%%%%%%%%
\subsection{Contribution}\label{sec:contribution}

When constrained CPSs are of interest, existing detection and mitigation strategies are developed assuming that an accurate knowledge of the plant's dynamic model is available. 
Consequently, recent progress made in the field of data-driven control, see, e.g.,  \cite{hou2013model,de2019formulas,robustallgower2023,Berberichdata2020}, have not been fully exploited to design data-driven anomaly detectors or resilient control strategies for constrained CPSs. Moreover, existing data-driven solutions, see, e.g., \cite{datazhao2023} and references therein, address CPS security issues only for unconstrained systems. 
The solution proposed in this paper goes in the direction of filling the identified gap. In particular, it proposes a novel data-driven implementation of the model-based solution presented in \cite{gheitasi2021safety} for constrained CPSs. The here developed solution leverages the data-driven method developed in \cite{alanwar2021data}, which is capable of computing zonotopic outer approximations of forward reachable sets to design a novel and robust data-driven passive anomaly detector local to the controller. Moreover, the data-driven set-theoretic predictive controller developed in \cite{attar2023data} is customized to design an emergency controller local to the plant, which is activated whenever the received control input is deemed unsafe according to the expected one-step plant evolution. 
%
%%%%%%%%%%%%%%%% SECTION %%%%%%%%%%%%%%%%%%%%%%%
\section{preliminaries and problem formulation}\label{section:preliminaries_and_problem_formulation}

\definition \label{def:polytope}
Given $q$ halfspaces, a  polytope $\mathcal{P}$ is defined as (using the $\mathcal{H}$-representation)
\begin{equation}\label{eq:polytope_def}
    \mathcal{P} = \left\{ x \in \rr^n | Cx \leq d, C\in \rr^{q \times n},d\in \rr^{q \times 1}   \right \}   
\end{equation}
\definition \label{def:zonotope}
Given a center vector $c \in \rr^n$ and $p \in \mathbb{N}$ generator vectors $g^{(i)}\in \rr^n$  collected in a matrix 
$G= \left[g^{(1)}\, \ldots,\, g^{(p)}\right]  \in \rr^{n \times p},$ referred to as the  generator matrix. Then, a zonotope is defined as (using the $\mathcal{G}-$representation)
\begin{equation}\label{eq:zonotope_def}
   \!\! \mathcal{Z}(c, G)\! =\! \! \left\{\! x\in\rr^n\!: x\! =\! c\! +\! \sum_{i=1}^{p} \beta^{(i)}g^{(i)}, -1\!\leq\! \beta^{(i)}\! \leq \! 1\! \right\}
\end{equation}
\definition \label{def:matrix_zonotope} %(\textbf{Matrix Zonotope}) 
Given a center matrix $C \in \rr^{n \times p}$ and $q\in \mathbb{N}$ generator matrices $G_M^{(i)}\in \rr^{n\times p}$  collected in a matrix  $G_{M} = \left[G^{(1)}_{M},\, \ldots,\, G^{(q)}_{M}\right]
\in \rr^{n \times (pq)}$. Then, a matrix zonotope is 
\begin{equation}\label{eq:matrix_zonotope}
	\begin{array}{rc}
		 \mathcal{M}(C, G_{M}) \!= \!&\displaystyle \!\!\!\!\! \{X \!\! \in \!\rr^{n \times p}\!:\!X\! =\! C +\! \sum_{i=1}^{q} \beta^{(i)}G^{(i)}_{M}, \vspace{-0.1cm}\\
		 & -1\leq \beta^{(i)} \leq 1 \}
	\end{array}   
\end{equation}
\definition \label{def:matrix_polytope} 
Consider a set of $n_v>0$ vertex matrices $\mathcal{V}_P=\{V_P^{(i)}\}_{i=1}^{n_v},$ $V_P^{(i)}\in \rr^{n\times p}.$ A matrix polytope 
$\displaystyle\mathcal{M}_P(\mathcal{V}_P)\!=\!
  \{M\in \rr^{n\times p}: M\!=\!\sum_{i=1}^{n_v}\rho_iV_P^{(i)},
  \,0\leq \rho_i\leq 1,\,  \sum_{i=1}^{n_v}\rho_i=1\}.$

\begin{definition}
Let $\mathcal{M}_1,\mathcal{M}_2$ be two matrix sets (e.g., matrix zonotopes) and $\mathcal{Z}$ a vector set (e.g., a zonotope). Then, %$\mathcal{M}_1\mathcal{M}_2$ and $\mathcal{M}_1\mathcal{Z}$ are defined as:
\begin{equation}
\begin{array}{c}
\mathcal{M}_1\mathcal{M}_2=\{M_1M_2: M_1\in \mathcal{M}_1,\,M_2\in \mathcal{M}_2\}
\\
\mathcal{M}_1\mathcal{Z}=\{M_1z: M_1\in \mathcal{M}_1,\,z\in \mathcal{Z}\}
\end{array}
\end{equation}
\end{definition}

Consider the class of systems described by a linear time-invariant (LTI) model subject to bounded but unknown disturbance. By denoting with  $k\in \mathbb{Z}_+ = \{0, 1, ...\}$ the discrete-time index, and with $x_k\in \rr^n,$  $u_k \in \rr^m,$ and $w_k \in \rr^n$  the state, control, and disturbance vectors, respectively, the discrete-time LTI evolution is
\begin{equation}\label{eq:linear_system}
    x_{k+1} = Ax_k + Bu_k + w_k
\end{equation}
where  $A, B$ are the system matrices of appropriate dimensions and $w_k$ is bounded in a compact set $\mathcal{W}\subset\rr^n.$  Assume that for physical limitations and safety reasons, the following set-membership constraints must be fulfilled:
\begin{equation}
x_k \in \mathcal{X}\subset \rr^n,\quad u_k \in \mathcal{U}\subset \rr^m,\quad  \label{eq:constraints}
\end{equation}
with $\mathcal{X}\subset \rr^n,\,\mathcal{U}\subset \rr^m, \mathcal{W}$ zonotopes containing the origin and described using the  $\mathcal{G}-$ or $\mathcal{H}-$ representation:
\begin{equation}\label{eq:constraints_h_representations}
\begin{array}{c}
\mathcal{X} = \mathcal{Z}_x(c_x,G_x) =\left\{x\! \in\! \rr^n\!:\! H_xx \leq h_x \right\}
\\
\mathcal{U} = \mathcal{Z}_u(c_u,G_u) =\{u\!\in\! \rr^m\!:\! H_uu\leq h_u\} \\
\end{array}
\end{equation}
\begin{equation}\label{eq:disturbance_zonotope}
 \mathcal{W}=\mathcal{Z}_w(c_w,G_w)  
 =
 \{w\in \rr^n: H_w w\leq h_w\}
\end{equation}
where $c_x,c_w\in \rr^n,c_u\in \rr^m,$ $G_x\in \rr^{n\times p_n},G_u\in \rr^{m\times p_m}, G_w\in \rr^{n\times p_w},$ $H_x\in \rr^{n_x\times n}, H_u\in \rr^{n_u\times m}, H_w\in \rr^{n_w\times n},$ $h_x\in \rr^{n_x}, h_u\in \rr^{n_u}, h_w\in \rr^{n_w}.$

\definition \label{def:mikowski_sum_and_difference}
Given two sets $\mathcal{S}_1$ and $\mathcal{S}_2$, the Minkowski/Pontryagin set sum (denoted as $\oplus$) and difference (denoted as $\ominus$) between $\mathcal{S}_1$ and $\mathcal{S}_2$ are: 
\begin{equation}\label{eq:def_for_mink_sum_and_diff}
\begin{array}{rcl}
%\mathcal{S}_1 \bigcup \mathcal{S}_2 &=& \{s: s \in \mathcal{S}_1 \text{ or }  \mathcal{S}_2 \}\\
\mathcal{S}_1 \oplus \mathcal{S}_2 &=& \{s_1 + s_2: s_1 \in \mathcal{S}_1 , s_2 \in \mathcal{S}_2 \}\\
\mathcal{S}_1 \ominus \mathcal{S}_2 &=& \{s_1 \in \rr^s: s_1 + s_2 \in \mathcal{S}_1, \forall s_2 \in \mathcal{S}_2 \}
\end{array}
\end{equation}
\definition \label{def:RCI_set}
A set $\mathcal{T}^0 \subseteq \mathcal{X}$ is called Robust Control Invariant (RCI)  for \eqref{eq:linear_system}-\eqref{eq:disturbance_zonotope} if $\forall x \in \mathcal{T}^0, \exists u \in \mathcal{U}: Ax + Bu + w \in \mathcal{T}^0, \forall w \in \mathcal{W}.$
\definition \label{def:model_based_controllable_sets}
Consider \eqref{eq:linear_system}-\eqref{eq:disturbance_zonotope} and a target set $\mathcal{T}^{j} \subseteq \mathcal{X}.$ The set of states Robust One-Step Controllable (ROSC) to $\mathcal{T}^{j}$ is
\begin{equation}\label{eq:ROSC-set}
\!\mathcal{T}^{j+1}\! =\! \{x\! \in\! \mathcal{X}: \exists u \in \mathcal{U}:\!\! Ax\! + Bu \!+ w \in \mathcal{T}^{j}, \forall w \in \! \mathcal{W} \}
\end{equation}

\definition \label{def:safety_def} The constrained system \eqref{eq:linear_system}-\eqref{eq:disturbance_zonotope} is said \textit{safe} if constraints \eqref{eq:constraints} are fulfilled $\forall\,k \geq 0$. 
%%%%%%%%%%%%%%%%% SECTION %%%%%%%%%%%%%%%%%%%%%%%%%
\subsection{Networked setup and problem formulation}\label{sec:system_under_attack}

Of interest are networked control system setups
where the plant and the tracking controller are spatially distributed, and the communication channel is subject to additive False Data Injection (FDI) attacks. Consequently, the closed-loop evolution of \eqref{eq:linear_system} is 
\begin{equation}\label{eq:system_under_attack}
       x_{k+1} = Ax_k + Bu^{'}_k + w_k, \qquad 
       u_k = \eta(x^{'}_k,r_k)
\end{equation}
where $r_k \in \rr^r$ is the reference signal and 
$\eta(\cdot,\cdot): \rr^n \times \rr^r \rightarrow \rr^m$ is the networked tracking controller logic. Moreover, $u^{'}_k:=u_k + u^a_k, x^{'}_k:=x_k + x^a_k$, with $u^a_k\in \rr^m$ and $x^a_k\in \rr^n$ the vectors injected by the attacker.

\assumption \label{assumption:control_center} In the absence of attacks, the networked tracking controller $u_k = \eta(x^{'}_k,r_k)$ ensures that the plant's constraints \eqref{eq:constraints} are fulfilled regardless of any admissible disturbance \eqref{eq:disturbance_zonotope} realization. The controller's working region, also known as controller Domain of Attraction (DoA), is hereafter denoted with $\mathcal{X}_{\eta} \subseteq\mathcal{X}.$
\assumption \label{assumption:problem_statement} 
The  matrices $A,$ $B$ of \eqref{eq:linear_system} are unknown.
Moreover, a collection of $N_t>0$ input-state trajectories is available,
\begin{equation}\label{eq:available_trajectories}
	\left\{\left\{u^{(i)}_k\right\}^{N_{s}^{(i)}-1}_{k=0}\!\!\!\!,\,  \left\{x^{(i)}_k\right\}^{N_{s}^{(i)}-1}_{k=0}\right\}_{i=1}^{N_t},
\end{equation}
where $N_{s}^{(i)}>0$ is the number of samples in each trajectory. 
Moreover, the matrix $\begin{bmatrix}
X_-^T & U_-^T
\end{bmatrix}^T$ has full row rank, i.e., 
\begin{eqnarray}
&& \text{rank}(\begin{bmatrix}
X_-^T & U_-^T
\end{bmatrix}^T)=n+m \label{eq:rank_condition}\\
\!\!\!X_{-}\!\!\!&\!\!\!=\!\!&\!\!\! \left[
    x^{(1)}_0, \cdots, x^{(1)}_{N_{s}^{(1)}-1}, \cdots, x^{(N_t)}_0, \cdots, x^{(N_t)}_{N_{s}^{(N_t)}-1}\right] \\
\!\!\!U_{-}\!\!\!&\!\!\!=\!\!&\!\!\! \left [
    u^{(1)}_0, \cdots, u^{(1)}_{N_{s}^{(1)} -1}, \cdots, u^{(N_t)}_0, \cdots, u^{(N_t)}_{N_{s}^{(N_t)} -1}
    \right ] 
\end{eqnarray}
and $X\in \rr^{n \times (N_s+1)N_t},$ $X_{-}\in \rr^{n \times N_sN_t},$  and  $U_{-}\in \rr^{m \times N_sN_t}.$
\remark \label{remark:data_driven_representation}The condition~\eqref{eq:rank_condition} is standard in the related literature \cite{de2019formulas}, and it ensures that the collected data \eqref{eq:available_trajectories} have been obtained for sufficiently persistent exciting input sequences and that $\begin{bmatrix}
    A & B
    \end{bmatrix} =  X_+ \begin{bmatrix}
    X_- \\ U_-
    \end{bmatrix}^{\dagger}$, with 
\begin{equation}\label{eq:compute_AB_without_noise}
\begin{array}{c}
    X_{+} =\left[
    x^{(1)}_1, \cdots, x^{(1)}_{N_{s}^{(1)}}, \cdots, x^{(N_t)}_1, \cdots, x^{(N_t)}_{N_{s}^{(N_t)}}\right ]
\end{array}
\end{equation}
and $\dagger$ denotes the right pseudo inverse operator. \hfill $\Box$

\noindent \textbf{Problem of interest:} {\it 
Under the Assumptions~\ref{assumption:control_center}-\ref{assumption:problem_statement}, design a data-driven networked control architecture such that 
\begin{itemize}
    \item \textit{(O1)} - FDI occurrences are detected before they could lead to a risk for the plant's safety (see Definition~\ref{def:safety_def}).
    \item \textit{(O2)} - The plant's safety under any FDI attack occurrence is preserved, and tracking performance recovered in the post-attack phase is ensured.
\end{itemize}
}
%
%%%%%%%%%%%%%%%%%%%%%%%%%%%%%%%%%%%%%%%%%%%%%%%%%
\section{Proposed Data-Driven Control Architecture}\label{sec:proposed_control_architecture}

\begin{figure}[h!]
    \centering
    \includegraphics[width=0.8\linewidth]{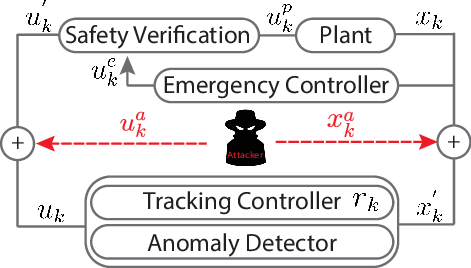}
    \caption{Proposed control architecture}
    \label{fig:proposed_control_architecture}
\end{figure}

To provide a solution to \textit{(O1)}-\textit{(O2)}, the networked data-driven control architecture shown in  Fig.~\ref{fig:proposed_control_architecture} is hereafter developed, where\\
    \noindent - the \textit{Anomaly Detector} module is a passive device, local to the tracking controller, in charge of detecting anomalies caused by FDI attacks\\
    \noindent - the \textit{Safety Verification} module is local to the plant, and it performs safety verification checks on the plant's safety. According to the outcome of these tests, it decides to either use the control signal received by the networked controller or the one provided by the local emergency controller.\\
    \noindent - the \textit{Emergency Controller} module is local to the plant, and it aims to confine the state trajectory in a neighbor of an offline defined safe equilibrium point.
\begin{remark}
In an attack-free scenario, the emergency controller cannot replace the networked controller because it is unaware of  $r_k.$ Its activation temporarily pauses the tracking task to preserve safety. \hfill $\Box$
\end{remark}
%
%%%%%%%%%%%%%%%% SECTION %%%%%%%%%%%%%%%%%%%%%%%%%
\subsection{Anomaly detector module} \label{sec:anomaly_detector}

If the model of the plant is available, a robust model-based passive anomaly detector can be designed using robust reachability arguments. In particular, if the model \eqref{eq:linear_system}-\eqref{eq:disturbance_zonotope} is available, then a binary anomaly detector can exploit the knowledge of the robust one-step reachable set starting from the current state $x_k$ under $u_k$ for any admissible disturbance realization. Consequently, the detector can claim an anomaly/attack when
\begin{equation}
 x'_{k+1} \not\in \mathcal{R}^+_k,\quad \mathcal{R}^+_k= Ax'_k \oplus Bu_k \oplus \mathcal{W} 
\end{equation}
However, in the considered data-driven setup, the set $\mathcal{R}^+_k$ cannot be directly computed, but it needs to be over-estimated from the available set of data \eqref{eq:available_trajectories}. The data-driven computation of an outer approximation of $\mathcal{R}^+_k$ can be obtained by adapting the approach developed in \cite{koch2021provably, alanwar2021data}. In particular, first, Lemma~\ref{lemma:noise_zonotope} describes the set of linear models $[\hat{A},\, \hat{B}]$ that are consistent with the data \eqref{eq:available_trajectories} and disturbance set $\mathcal{Z}_{w}.$ Then, Lemma~\ref{lemma:over_approx_forward_sets}, exploits the results of Lemma~\ref{lemma:noise_zonotope} to build an outer approximation, namely $\hat{\mathcal{R}}_k^+$, of $\mathcal{R}^+_k$.

\begin{lemma}{}\label{lemma:noise_zonotope} \it \cite[Lemma 1]{alanwar2021data}
Let $T =\displaystyle \sum_{i=1}^{N_t} N^{(i)}_s$ and consider the following concatenation of multiple noise zonotopes 
$$\mathcal{M}_w=\mathcal{M}_w(C_w, [G^{(1)}_{{M}_w}, \ldots, G^{(qT)}_{{M}_w}]),
$$ 
where
$C_w\in \rr^{n\times (n+m)}=[c_{w},\ldots,\,c_w]$, and 
$G_{M_w}\in \rr^{n\times T(n+m)}$ is built $\forall\,i \in \{ 1, \ldots, q\},\, \forall\, j \in \{2, \ldots, T-1\}$ as 
\begin{equation}
    \begin{array}{rcl}
        G^{(1+(i-1)T)}_{{M}_w} &=& \begin{bmatrix}
                                 g^{(i)}_{w} & 0_{n \times (T-1)}
                                \end{bmatrix}
        \\
        G^{(j+(i-1)T)}_{{M}_w} &=& \begin{bmatrix}
                                0_{n \times (j-1)} & g^{(i)}_{w} & 0_{n \times (T-j)}
                                \end{bmatrix}
        \\
        G^{(T+(i-1)T)}_{{M}_w} &=& \begin{bmatrix}
                                0_{n \times (T-1)} & g^{(i)}_{w}
                                \end{bmatrix}
    \end{array}
\end{equation}
Then, the matrix zonotope
\begin{equation}\label{eq:compute_Mzono_AB}
\begin{array}{rcl}
   
\mathcal{M}_{{A} {B}}\!\!\!\!\!\!&=&\!\!\!\!\!\! (X_+ - \mathcal{M}_w) \begin{bmatrix}
    X_- \\ U_-
    \end{bmatrix}^{\dagger}   
    \\
     \!\!\!& := &\!\!\!\!\!\! \{[\hat{A},\, \hat{B}]: 
     [\hat{A},\, \hat{B}]\! =\! C_{AB} + \displaystyle\sum_{i=1}^{T} \beta^{(i)}G^{(i)}_{M_{AB}},\\
    &&
         -1\leq \beta^{(i)} \leq 1 \}
\end{array}
\end{equation}
where
$$
\begin{array}{c}
C_{AB} = (X_{+}-C_{w})\left(
    [X^T_{-},\, U^T_{-}]    ^{T}\right)^{\dagger}\\
    G_{M_{AB}} =\left[
        G^{(1)}_{M_w}
    \left(
    [X^T_{-},\, U^T_{-}]    ^{T}\right)^{\dagger}\!\!, \ldots, G^{(qT)}_{M_w}\left(
    [X^T_{-},\, U^T_{-}]    ^{T}\right)^{\dagger}
    \right]
\end{array}
$$
contains the set of all system matrices $[\hat{A},\, \hat{B}]$ consistent with \eqref{eq:available_trajectories} and $\mathcal{Z}_{w}$ and such that $[A,B]\in \mathcal{M}_{AB}.$ $\hfill\square$
\end{lemma}
\begin{lemma} (Adapted from \cite[Theorem 1]{alanwar2021data})\label{lemma:over_approx_forward_sets}
\it 
   The set $\hat{\mathcal{R}}^+_k\subset \rr^n,$ computed as 
    \begin{equation} \label{eq:over_approximation_of_one_step_evolution}
      \hat{\mathcal{R}}^+_k = \mathcal{M}_{AB}[x'_k,u_k]^T \oplus \mathcal{W}  
    \end{equation}
    is an outer approximation of $\mathcal{R}^+_k$. 
\end{lemma}

Given, the result of Lemma~\ref{lemma:over_approx_forward_sets}, the proposed robust data-driven anomaly detector is
\begin{equation} \label{eq:anomaly_detector_logic}
    D_k =  \begin{cases}
       \text{anomaly} & \text{ if } x^{'}_{k+1} \not\in \hat{\mathcal{R}}^+_k \\
       \text{normal} & \text{ otherwise}
    \end{cases}
\end{equation}

\begin{remark}
The detector logic \eqref{eq:anomaly_detector_logic} is robust against any disturbance realization, and it ensures the absence of false positives. Such a feature is important to avoid unnecessary activation of the emergency controller with consequent tracking performance losses. \hfill $\Box$
\end{remark}

In what follows, we assume that when the anomaly detector triggers an anomaly, a flag=1 is transmitted.

%%%%%%%%%%%%%%%% SECTION %%%%%%%%%%%%%%%%%%%%%%%%%%%
\subsection{Safety verification module} \label{sec:safety_verification}

The objective of this module is to prevent the plant from reaching unsafe configurations, i.e., a configuration where  $x_k \notin \mathcal{X}_{\eta}\subseteq \mathcal{X}$ or $ u'_k \notin \mathcal{U}$.  In particular, given the received $u_k',$ the safety modules checks the following anomalies:
\begin{enumerate}
    \item \textbf{If} flag ==1 \textbf{then} \eqref{eq:anomaly_detector_logic} detected an attack.
    \item \textbf{If} $u_k'\notin \mathcal{U}$ or $\hat{\mathcal{S}}^{+}_k \not \subseteq \mathcal{X}_{\eta},$ with \begin{equation}
\hat{\mathcal{S}}^{+}_k = \mathcal{M}_{AB}[x_k, u_k']^T \oplus \mathcal{W}
\end{equation}
\textbf{then} the received control input has been corrupted, and it might cause constraints violations.
\end{enumerate}

The condition $\hat{\mathcal{S}}^{+}_k \not \subseteq \mathcal{X}_{\eta}$ ensures that any attack will be detected at least one step before it could compromise the plant's safety. Indeed,  if  $\hat{\mathcal{S}}^{+}_k \not \subseteq \mathcal{X}_{\eta},$ there is a possibility that $x_{k+1}$ caused by $u_k'$ does not fulfill the constraints. 
Moreover, in both cases above, the control input $u_k'$ received by the networked controller cannot be applied.

%%%%%%%%%%%%%%%%% SECTION %%%%%%%%%%%%%%%%%%%%%%%%%%%
\subsection{Emergency controller module} \label{sec:emergency_controller}

The objective of the emergency controller is to stabilize the plant around a pre-established equilibrium point whenever the control input received by the networked tracking controller cannot be trusted for safety reasons. 
By assuming $(0_n,0_m)$ as the equilibrium pair, 
the emergency controller 
\begin{equation}
    u_k^{e}=f_e(x_k),\quad f_e: \mathcal{X}_e \in  \rr^n \rightarrow \mathcal{U}_e \in \rr^m 
\end{equation}
must fulfill the following objectives: (i) {\it $\mathcal{X}_e\supseteq \mathcal{X}_{\eta}$ and $\mathcal{U}_e\subseteq \mathcal{U}$}; (ii) {\it there exists an RCI set $\hat{\mathcal{T}}_e^0 \subseteq \mathcal{X}_{\eta}$  where the state trajectory is uniformly ultimately bounded in a finite number of steps}. Condition (i) ensures that the emergency controller fulfills all the constraints and that it can be activated anytime and from any state reachable under the tracking controller, while condition (ii) guarantees that the networked controller can be safely re-activated at least when $x_k\in \hat{\mathcal{T}}_e^0.$  

To meet the objectives (i)-(ii), in what follows, we resort to a customization of the data-driven set-theoretic MPC controller proposed in \cite{attar2023data}. The controller's offline design steps can be summarized as follows:\

    1) Design a stabilizing data-driven terminal controller, namely $u_k^{e_0}=f_e^0(x_k),$ fulfilling the prescribed constraints and such that the associated RCI set $\hat{\mathcal{T}}_e^0$ is a subset of $\mathcal{X}_{\eta},$ i.e., $\hat{\mathcal{T}}_e^0\subseteq \mathcal{X}_{\eta}.$ 
    
    2) If $\hat{\mathcal{T}}_e^0 \subseteq \mathcal{X}_{\eta},$ recursively compute a family $\{\hat{\mathcal{T}}^j_e\}^N_{j=1}$ of $N>0$ sets,
     with  $\hat{\mathcal{T}}^j_e$ an inner approximation of the model-based ROSC set defined in \eqref{eq:ROSC-set} and $\bigcup_{j=0}^N\hat{\mathcal{T}}_e^j \supseteq \mathcal{X}_{\eta}.$

The RCI set $\hat{\mathcal{T}}_e^0$ can be computed using the solution in \cite{chen2021data, RCI_LPV2023} or \cite[Remark~5]{attar2023data}, while  $\{\hat{\mathcal{T}}^j_e\}^N_{j=1},$ can be computed using the augmented description proposed in  \cite[Sec.~III.C]{attar2023data}. In particular,
\begin{equation}\label{eq:inner_ROSC_data_driven_augm}
		\begin{array}{c}
  \hat{\mathcal{T}}^{j}_e= Proj_x(\hat{\Xi}^j_e)=\left\{\!x\! \in\! \rr^n\!:\! H_{\hat{\mathcal{T}}^{j}_e} x \leq h_{\hat{\mathcal{T}}^{j}_e}  \!\right\}, \\
        \hat{\Xi}^{j}_e \!=\text{In}_z\left\{\hat{\Xi}^{j}_{AB}\right\},
		\\
		\!\hat{\Xi}^{j}_{AB} \!= \!\!\!\!\!\! \displaystyle  
		\bigcap_{[\hat{A}_i,\hat{B}_i]\in \mathcal{V}_{AB}}\!\!\!\!\!\!\!\!\! 
		\left\{z=[x^T,u^T]^T\! \in \rr^{n+m}\!:\!H_{z}^iz\leq h_{z}^i\right\} 
	\end{array}
\end{equation}
where $\hat{\Xi}^{j}_{AB}$ is the $(x,u)-$ augmented space description of  the ROSC set, $\text{In}_z(\cdot)$ is an operator which computes a zonotopic inner approximation of a polytope (e.g., using \cite[Sec.~IV.A.2]{yang2021scalable}),  $Proj_x(\hat{\Xi}^j_e)$ performs a projection operation of  $\hat{\Xi}^j_e$ into the $x-$domain, $\mathcal{V}_{AB}$ is the set of vertices ${\hat{A}_i,\,\hat{B}_i}\in \mathcal{M}_{AB},$ and 
\begin{equation}\label{eq:H-rep_extended}
	H_z^i= \begin{bmatrix}
		H_x & 0 \\
		H_{\hat{\mathcal{T}}^{j-1}_e}\hat{A}_i & H_{\hat{\mathcal{T}}^{j-1}_e}\hat{B}_i\\
		0 & H_u
	\end{bmatrix}, 
	\quad
	h_z^i=\begin{bmatrix}
		h_x\\
		\Tilde{h}_{\hat{\mathcal{T}}^{j-1}_e}\\
		h_u
	\end{bmatrix} 
\end{equation}
with 
\begin{equation}\label{eq:compute_tilde_h}
	[\tilde{h}_{\hat{\mathcal{T}}^{j-1}_e}]_r = \min_{w\in \mathcal{W}}\left\{[h_{\hat{\mathcal{T}}^{j-1}_e}]_r - [H_{{\hat{\mathcal{T}}^{j-1}_e}}]_r w \right\}  
\end{equation}
and $[h_{\hat{\mathcal{T}}^{j-1}_e}]_r,$ $[H_{{\hat{\mathcal{T}}^{j-1}_e}}]_r$ the $r-th$ rows of $h_{\hat{\mathcal{T}}^{j-1}_e}$ and $H_{{\hat{\mathcal{T}}^{j-1}}_{e}}.$
\remark \label{remark:zonotopic_approx_rosc}
The inner zonotopic approximation of $\hat{\Xi}^{j}_{AB}$ is needed to efficiently compute the projection $\hat{\mathcal{T}}_e^j$ from $\hat{\Xi}_{AB}^j,$ see \cite[Remark~4]{attar2023data} for a detailed discussion. \hfill $\Box$

Given $\{\hat{\mathcal{T}}^j_e\}^N_{j=0},$ and a convex cost function $J(x_k,u),$ the online operations of the Emergency Data-driven Set-Theoretic Controller (E-DSTC) are collected in Algorithm~\ref{algorithm:set_theoretic_MPC}.

\begin{algorithm}[h!] 
\textbf{Input:}
RCI set ${\mathcal{T}}^0_e$, and ROSC sets $\{\hat{\Xi}^{j}_e$, $\hat{\mathcal{T}}^{j}_e\}_{i=1}^N$ \;

\noindent
	\xrfill[0.7ex]{1pt} ($\forall\,k$)\xrfill[0.7ex]{1pt} 
\begin{algorithmic}[1]
 \State Find set membership index $j_k := \displaystyle \!\!\!\!\!\! \min_{j\in \{0,\ldots,N\}} \{\!j\!: \! x_k \in \! \hat{\mathcal{T}}^j_e \!\}$\;
\If{$j_k=0$}{ $u^{e}_k = f_e^0(x_k)$} \textbf{else}
%\Else{ compute $u_k^e$ by solving the following QP problem
\begin{equation}\label{eq:control_based_extended_control_regions}
	\begin{array}{c}
\displaystyle		u^{e}_k = \arg\min_u J(x_k, u) \quad s.t. 
	\left[x_k^T,u^T\right]^T \in \hat{\Xi}^j_e
	\end{array}
\end{equation}
%}
\EndIf
\end{algorithmic}
\caption{Emergency Controller (E-DSTC)}
\label{algorithm:set_theoretic_MPC}
\end{algorithm}
\begin{property}
For any $x_0\in  \bigcup_{j=0}^N\hat{\mathcal{T}}^j_e,$ the emergency controller E-DSTC ensures that the state trajectory is uniformly ultimately bounded in $\hat{\mathcal{T}}^0_e$ in at most $N-$steps, i.e., $x_{k}\in \hat{\mathcal{T}}^0_e,\,\forall k\geq N$ \cite{attar2023data}.    
\end{property}

\subsection{Switching control logic}\label{sec:switching_control}

According to the anomaly test performed by the safety verification module  (see subsection~\ref{sec:safety_verification}), the control input applied to the plant, namely $u_k^p,$ is either $u_k'$ (provided by the networked controller) or $u_k^e$ (provided by the emergency controller). In this regard, Algorithm~\ref{algorithm:switching_logic} describes the policy used to switch from $u_k'$ to $u_k^e$ and vice-versa.

\begin{algorithm}[h!] 
\textbf{Initialize:} emergency = false %{\color{red}ignore = 0}

\noindent
\xrfill[0.7ex]{1pt} ($\forall\,k$)\xrfill[0.7ex]{1pt}
\begin{algorithmic}[1]
\If{(emergency==true $\land$ $x_k \! \in \! \hat{\mathcal{T}}^0_e$)}{ emergency=false, ignore=1} \label{step:stay_in_emergency} \textbf{else} ignore=0
%\Else { }ignore=0
\EndIf
\If{$(u_k'\notin \mathcal{U})$  $\lor$
$\left(\hat{\mathcal{S}}^{+}_k\not \subseteq \mathcal{X}_{\eta} \right)$ $\lor$ (flag==1 $\land$  ignore==0)}{ emergency = true} \label{step:safety_verification_module}
\EndIf
\If{emergency == true}{ $u^{p}_k = u^{e}_k$} \label{step:emergency_activated} \textbf{else} $u^{p}_k = u'_k$
%\Else{ $u^{p}_k = u'_k$}
\EndIf
\end{algorithmic}
\caption{Switching control policy}
\label{algorithm:switching_logic}
\end{algorithm}

\begin{proposition} \label{proposition:complete_architecture}
\it 
The E-DSTC emergency controller and the switching control policy (Algorithm~\ref{algorithm:switching_logic}) allow the control architecture in Fig.~\ref{fig:proposed_control_architecture} to preserve the safety of the plant under any cyber-attacks and performance recovery in the post-attack phase.
\end{proposition}

\begin{proof}
In the worst scenario, the anomaly detector \eqref{eq:anomaly_detector_logic} can either not detect the attack, or the anomaly flag=1 is reverted to flag=0 by an FDI attack on the actuation channel. In both cases, the control input $u'_k$ received on the plant side is corrupted. However, the test performed by the safety verification module in Step~\ref{step:safety_verification_module} of Algorithm~\ref{algorithm:switching_logic} verifies if $u'_k$ is admissible for the given input constraint and if the associate worst-case one step evolution remains in the tracking controller's domain $\mathcal{X}_{\eta}.$ If a safety risk is detected (emergency=true), then the E-DSTC controller is activated (Step~\ref{step:emergency_activated}). Note that, by construction, $\mathcal{X}_{\eta} \subseteq \bigcup_{j=0}^N{\hat{\mathcal{T}_e^j}},$ and the switch $u_k'\rightarrow u_k^e$ can be performed fulfilling all the constraints $\forall\, x\in \mathcal{X}_{\eta}.$ Moreover, once E-DSTC is activated, then Step~\ref{step:stay_in_emergency} prescribes that such a controller remains active until the state trajectory is confined (in at most $N$ steps) into $\hat{\mathcal{T}}_0^j.$ Then, when $x_k\in \hat{\mathcal{T}}_0^j$, two cases can arise: an attack is still ongoing, or there is a post-attack phase. In both cases, the proposed policy will attempt the switch $u_k^e\rightarrow u_k'$ under the pre-requisite that no safety risk could arise by applying $u_k'.$ Differently from the switch from $u_k'\rightarrow u_k^e$, Step~\ref{step:safety_verification_module} is now instructed to ignore for only one iteration the anomaly detected possible raised by \eqref{eq:anomaly_detector_logic}. This is instrumental to taking into consideration the fact that the current state $x_k$ was obtained using $u_{k-1}^e$ and not $u_{k-1}$ as expected by the \eqref{eq:anomaly_detector_logic}. If an attack is still ongoing, then the emergency controller is kept active; otherwise, since $x_k\in \mathcal{X}_{\eta},$ the tracking controller can be safely reactivated, ensuring tracking performance recovery in the post-attack phase.
\end{proof}
%
%%%%%%%%%%%%%% SECTION %%%%%%%%%%%%%%%%%%%%%%%%%%%%%%
\section{Simulation Results} \label{sec:simulation}
This section shows the effectiveness of the proposed safety-preserving architecture through numerical results obtained considering the two-tank system described in \cite{gheitasi2022worst}. 
The system consists of two tanks, which water levels $h_i,\,i=1,2$ are the state variables, i.e., $x_k=[h_{1_k}, h_{2_k}]^T\in \rr^2.$ The input vector is $u = \left[u_p, u_l, u_u\right]^T \in \rr^3,$ where $u_p$ regulates the valve injecting water within the first tank, while $u_l$ and $u_u$ control the lower and upper valves between the two tanks. We assume that the system model is unknown but working around the equilibrium point $x_{eq}=\left[0.5, 0.5\right]^T$, $u_{eq}=\left[0.938, 1, 0.833\right]^T,$ and that data are collected with a sampling time $T_s = 1 \sec.$ By simulating the system around the considered equilibrium and applying random input perturbations, we have collected two input-state trajectories \eqref{eq:available_trajectories}
%, each containing $N_s=3$ samples 
which verifies the rank condition \eqref{eq:rank_condition}. From the collected data, the set of all system matrices $\mathcal{M}_{AB}$ has been computed as prescribed by Lemma~\ref{lemma:over_approx_forward_sets}.
Moreover, the constraint and disturbance sets around the equilibrium point are described by the constraints $-0.7778 \leq {u}_p \leq 0.6111, -1.25 \leq {u}_l \leq 0.75, -1.4765 \leq {u}_u \leq 0.5235$ and $-0.48 \leq h_1, h_2 \leq 0.3,$ $\mathcal{W}=\{w \in \rr^2: |w_j|\leq 0.001, j=1,2\}$. 

In the performed experiments, the networked tracking controller has been designed using the data-driven LQR design procedure proposed in \cite{de2019formulas}, where a saturation module has been added to ensure constraint fulfillment. On the other hand, the E-DSTC controller has been designed using $40$ ROSC sets $\{\hat{\mathcal{T}}_e^j\}^{40}_{j=1}.$ Moreover, $x_0=[0.01, -0.01]^T$ and $r_k=[0.1,\, 0.03]^T$ for $k\in [0,\,100)$ and $r_k=[0.1,\,-0.1]^T$ for $k\in [100,\,200].$
%
%%%%%%%%%%%%%%%%%% SECTION %%%%%%%%%%%%%%%%%%%%%%%%%%%
\subsection{Attack on the actuation channel} \label{sec:attack_on_acutuation}

In the first scenario, the attacker performs an intelligent FDI attack on the actuation channel to affect the plant tracking performance while avoiding triggering the safety checks performed by the safety verification module. In particular, for $k\in [95,\, 113],$ the attacker injects $u^a_k$ where $\displaystyle u^a_k=\arg\min_{u^a} \|u^a-[1,\,1,\,2]^T\|_2^2,\,\text{ s.t. } u^a+u_k\in \mathcal{U}$ and set flag=0.
The obtained results are shown in Figs.~\ref{fig:case_a_state_evolution} and \ref{fig:case_a_alarm_safety}. When the attack starts at $k=95,$ the state trajectory deviates from the desired reference (see the red line in Fig.~\ref{fig:case_a_state_evolution}). At $k=96,$ the anomaly detector \eqref{eq:anomaly_detector_logic} detects an anomaly (see Fig.~\ref{fig:case_a_alarm_safety}) because the received measurement $x_{96}'\notin \mathcal{R}^+_{95}.$ However, since the attacker corrupts the anomaly flag, then an emergency is not raised by the safety checks until $k=104$ when the plant is one-step away from possibly violating the constraints, i.e., $\hat{\mathcal{S}}^+_{104} \not \subseteq \hat{\mathcal{T}}^{40}_e$ (see small box in Fig.~\ref{fig:case_a_state_evolution}). Given the safety risk, at $k=104,$ the emergency controller is activated to steer the state of the system into the safety region $\hat{\mathcal{T}}^0_e$ without constraint violation (see the green line in Fig.~\ref{fig:case_a_state_evolution}). 
When at $k=115$ the state of the state trajectory reaches $\hat{\mathcal{T}}^0_e$, the networked tracking controller can be safely restored since the attack is terminated. Consequently, the plant can recover and track the desired $r_k$ (see the magenta line in Fig.~\ref{fig:case_a_state_evolution}).  
\begin{figure}[h!]
    \centering
    \includegraphics[width=1\linewidth]{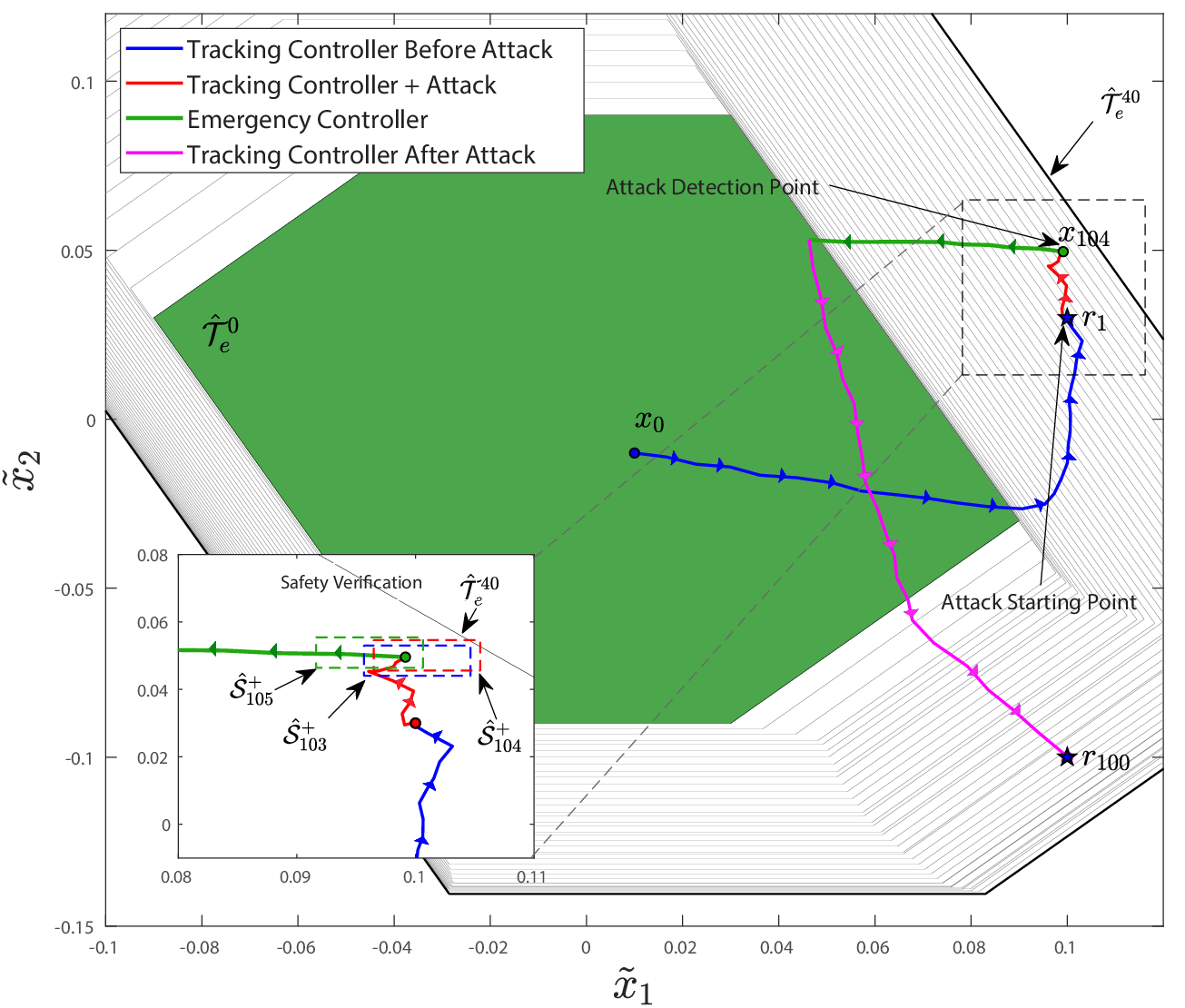}
    \caption{Case A: State Evolution and Emergency Controller Operation}
    \label{fig:case_a_state_evolution}
\end{figure}
\begin{figure}[h!]
    \centering
    \includegraphics[width=1\linewidth]{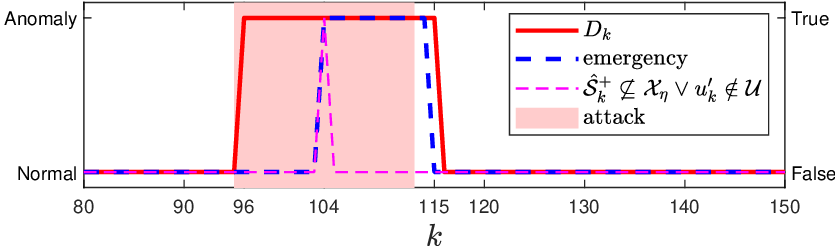}
    \caption{Case A: Detector and Safety Verification Outputs}
    \label{fig:case_a_alarm_safety}
\end{figure}
%
%%%%%%%%%%%%%%%% SECTION %%%%%%%%%%%%%%%%%%%%%%
\subsection{Attack on the measurement channel} \label{sec:attack_on_measurement}

In the second scenario, the attacker performs an FDI on the measurement channel. In particular, for $k\in [95,\,113],$ it corrupts $x_k$ by adding on it the vector $[0.0025(k-94), 0]^T.$  The obtained results are shown in Figs~\ref{fig:case_b_state_evolution}-\ref{fig:case_b_alarm_safety}.
Since at the beginning of the attack period, the magnitude of the injection is minimal and comparable with the one of the process disturbance $\mathcal{W}$, the attack is undetected by \eqref{eq:anomaly_detector_logic} until $k=109$ (see Fig.~\ref{fig:case_b_alarm_safety}). 
At $k=110$, the attack is detected by \eqref{eq:anomaly_detector_logic} because  $x_{110}' \notin \hat{\mathcal{R}}^+_{109}$ (see small box in Fig.~\ref{fig:case_b_state_evolution}). Consequently, the anomaly detector transmits flag=1 to the plant, triggering the activation of the emergency controller (see green line in Fig.~\ref{fig:case_b_state_evolution}). At $k=114$, under the E-DSTC, the system's state reaches $\hat{\mathcal{T}}^0_e$ while the attack is terminated. Consequently, at $k=114,$ the anomaly still raised by \eqref{eq:anomaly_detector_logic} is ignored for one step (see Step~\eqref{step:safety_verification_module} of Algorithm.~\ref{algorithm:switching_logic}) and the tracking controller is reactivated. For $k\geq 114$, the state trajectory resumes tracking $r_k$ (see the magenta line in Fig.~\ref{fig:case_b_state_evolution}).
\begin{figure}[h!]
    \centering
    \includegraphics[width=1\linewidth]{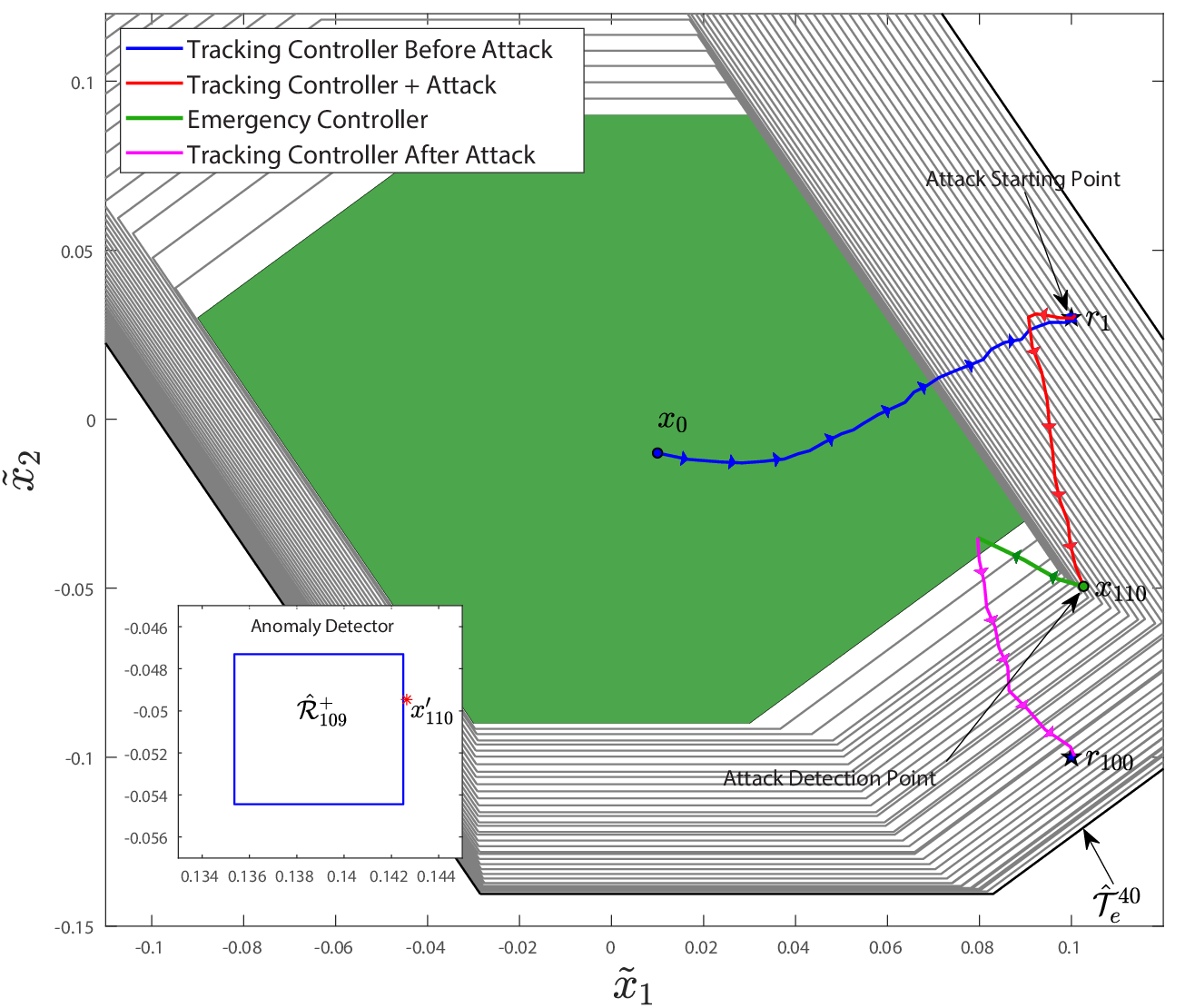}
    \caption{Case B: State Evolution and Detectors Operation}
    \label{fig:case_b_state_evolution}
\end{figure}
\begin{figure}[h!]
    \centering
    \includegraphics[width=1\linewidth]{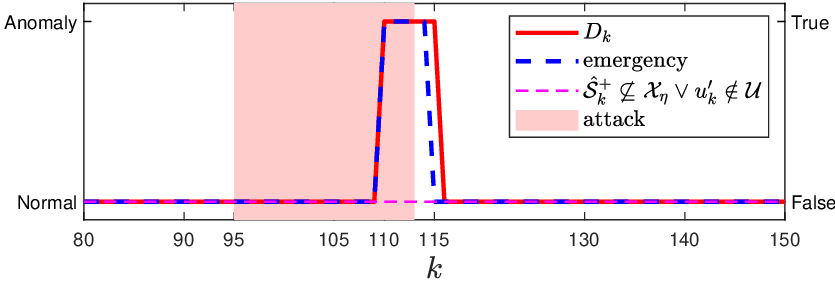}
    \caption{Case B: Detector and Safety Verification Outputs}
    \label{fig:case_b_alarm_safety}
\end{figure}
%
%%%%%%%%%%%%%%%% SECTION %%%%%%%%%%%%%%%%%%%%%%%%%%
\section{Conclusions} \label{sec:conclusions}

In this paper, a data-driven safety-preserving control architecture for constrained CPS has been proposed. The proposed solution leveraged worst-case backward and forward reachability arguments to develop a data-driven robust anomaly detector on the networked controller and a safety-preserving module on the plant. Moreover, a safe switching mechanism between the networked and local emergency controllers has been defined to ensure safety during the attack and performance recovery thereafter. It has been proved that the proposed architecture ensures attack detection at least one-step before the plant's safety could be compromised. Simulation results have been shown to illustrate the capability of the proposed data-driven control architecture. 
%%%%%%%%%%%%%%%%%%%%%%%%%
\bibliographystyle{IEEEtran}
\bibliography{bibliography} 

\begin{thebibliography}{10}
\providecommand{\url}[1]{#1}
\csname url@samestyle\endcsname
\providecommand{\newblock}{\relax}
\providecommand{\bibinfo}[2]{#2}
\providecommand{\BIBentrySTDinterwordspacing}{\spaceskip=0pt\relax}
\providecommand{\BIBentryALTinterwordstretchfactor}{4}
\providecommand{\BIBentryALTinterwordspacing}{\spaceskip=\fontdimen2\font plus
\BIBentryALTinterwordstretchfactor\fontdimen3\font minus
  \fontdimen4\font\relax}
\providecommand{\BIBforeignlanguage}[2]{{%
\expandafter\ifx\csname l@#1\endcsname\relax
\typeout{** WARNING: IEEEtran.bst: No hyphenation pattern has been}%
\typeout{** loaded for the language `#1'. Using the pattern for}%
\typeout{** the default language instead.}%
\else
\language=\csname l@#1\endcsname
\fi
#2}}
\providecommand{\BIBdecl}{\relax}
\BIBdecl

\bibitem{dibaji2019systems}
S.~M. Dibaji, M.~Pirani, D.~B. Flamholz, A.~M. Annaswamy, K.~H. Johansson, and
  A.~Chakrabortty, ``A systems and control perspective of {CPS} security,''
  \emph{Annual reviews in control}, vol.~47, pp. 394--411, 2019.

\bibitem{ghaderi2020blended}
M.~Ghaderi, K.~Gheitasi, and W.~Lucia, ``A blended active detection strategy
  for false data injection attacks in cyber-physical systems,'' \emph{IEEE
  Transactions on Control of Network Systems}, vol.~8, no.~1, pp. 168--176,
  2020.

\bibitem{miao2016coding}
F.~Miao, Q.~Zhu, M.~Pajic, and G.~J. Pappas, ``Coding schemes for securing
  cyber-physical systems against stealthy data injection attacks,'' \emph{IEEE
  Trans. on Control of Network Systems}, vol.~4, no.~1, pp. 106--117, 2016.

\bibitem{attaractive}
M.~Attar and W.~Lucia, ``An active detection strategy based on dimensionality
  reduction for false data injection attacks in cyber-physical systems,''
  \emph{IEEE Trans. on Control of Network Systems}, pp. 1--11, 2023.

\bibitem{franze2023cyber}
G.~Franz{\`e}, D.~Famularo, W.~Lucia, and F.~Tedesco, ``Cyber--physical systems
  subject to false data injections: A model predictive control framework for
  resilience operations,'' \emph{Automatica}, vol. 152, p. 110957, 2023.

\bibitem{sun2019resilient}
Q.~Sun, K.~Zhang, and Y.~Shi, ``Resilient model predictive control of
  cyber--physical systems under dos attacks,'' \emph{IEEE Trans. on Industrial
  Informatics}, vol.~16, no.~7, pp. 4920--4927, 2019.

\bibitem{franze2023output}
G.~Franze, D.~Famularo, F.~Tedesco \emph{et~al.}, ``Output based resilient
  model predictive control of cyber-physical systems under cover attacks,'' in
  \emph{IFAC World Congress}, 2023.

\bibitem{lucia2022supervisor}
W.~Lucia, G.~Franzè, and B.~Sinopoli, ``A supervisor-based control
  architecture for constrained cyber-physical systems subject to network
  attacks,'' \emph{IEEE Trans. on Control of Network Systems}, vol.~10, no.~3,
  pp. 1184--1194, 2023.

\bibitem{gheitasi2022worst}
K.~Gheitasi and W.~Lucia, ``A worst-case approach to safety and reference
  tracking for cyber-physical systems under network attacks,'' \emph{IEEE
  Trans. on Automatic Control}, vol.~68, no.~7, pp. 4391--4397, 2023.

\bibitem{hou2013model}
Z.-S. Hou and Z.~Wang, ``From model-based control to data-driven control:
  Survey, classification and perspective,'' \emph{Information Sciences}, vol.
  235, pp. 3--35, 2013.

\bibitem{de2019formulas}
C.~De~Persis and P.~Tesi, ``Formulas for data-driven control: Stabilization,
  optimality, and robustness,'' \emph{IEEE Trans. on Automatic Control},
  vol.~65, no.~3, pp. 909--924, 2019.

\bibitem{robustallgower2023}
J.~Bongard, J.~Berberich, J.~Köhler, and F.~Allgöwer, ``Robust stability
  analysis of a simple data-driven model predictive control approach,''
  \emph{IEEE Trans. on Automatic Control}, vol.~68, no.~5, pp. 2625--2637,
  2023.

\bibitem{Berberichdata2020}
J.~Berberich, J.~Köhler, M.~A. Müller, and F.~Allgöwer, ``Data-driven model
  predictive control with stability and robustness guarantees,'' \emph{IEEE
  Trans. on Automatic Control}, vol.~66, no.~4, pp. 1702--1717, 2021.

\bibitem{datazhao2023}
Z.~Zhao, Y.~Xu, Y.~Li, Z.~Zhen, Y.~Yang, and Y.~Shi, ``Data-driven attack
  detection and identification for cyber-physical systems under sparse sensor
  attacks,'' \emph{IEEE Trans. on Automatic Control}, vol.~68, no.~10, pp.
  6330--6337, 2023.

\bibitem{gheitasi2021safety}
K.~Gheitasi and W.~Lucia, ``A safety preserving control architecture for
  cyber-physical systems,'' \emph{International Journal of Robust and Nonlinear
  Control}, vol.~31, no.~8, pp. 3036--3053, 2021.

\bibitem{alanwar2021data}
A.~Alanwar, A.~Koch, F.~Allgöwer, and K.~H. Johansson, ``Data-driven
  reachability analysis from noisy data,'' \emph{IEEE Trans. on Automatic
  Control}, vol.~68, no.~5, pp. 3054--3069, 2023.

\bibitem{attar2023data}
M.~Attar and W.~Lucia, ``Data-driven robust backward reachable sets for
  set-theoretic model predictive control,'' \emph{IEEE Control Systems
  Letters}, vol.~7, pp. 2305--2310, 2023.

\bibitem{koch2021provably}
A.~Koch, J.~Berberich, and F.~Allg{\"o}wer, ``Provably robust verification of
  dissipativity properties from data,'' \emph{IEEE Trans. on Automatic
  Control}, vol.~67, no.~8, pp. 4248--4255, 2021.

\bibitem{chen2021data}
Y.~Chen and N.~Ozay, ``Data-driven computation of robust control invariant sets
  with concurrent model selection,'' \emph{IEEE Trans. on Control Systems
  Technology}, vol.~30, no.~2, pp. 495--506, 2021.

\bibitem{RCI_LPV2023}
M.~Mejari, A.~Gupta, and D.~Piga, ``Data-driven computation of robust invariant
  sets and gain-scheduled controllers for linear parameter-varying systems,''
  \emph{IEEE Control Systems Letters}, doi: 10.1109/LCSYS.2023.3329291, 2023.

\bibitem{yang2021scalable}
L.~Yang and N.~Ozay, ``Scalable zonotopic under-approximation of backward
  reachable sets for uncertain linear systems,'' \emph{IEEE Control Systems
  Letters}, vol.~6, pp. 1555--1560, 2021.

\end{thebibliography}
\end{document}